\documentclass[twoside]{article}
\usepackage[accepted]{aistats2018}

\usepackage{natbib}
\setcitestyle{authoryear,open={(},close={)}}

\usepackage{amsmath}
\usepackage{amsthm}
\usepackage{amssymb}
\usepackage{algorithm}
\usepackage{subfig}
\usepackage{color}
\usepackage[english]{babel}
\usepackage{graphicx}
\usepackage{wrapfig,epsfig}
\usepackage{psfrag}
\usepackage{epstopdf}
\usepackage{url}
\usepackage{graphicx}
\usepackage{color}
\usepackage{epstopdf}
\usepackage{comment}
\usepackage{hyperref}
\hypersetup{colorlinks=true,citecolor=blue,linkcolor=blue}
\usepackage[T1]{fontenc}
\usepackage{verbatim}

\usepackage{tikz}
\usepackage{bbm}
\usepackage{algpseudocode}
\usepackage[hyphenbreaks]{breakurl}
\usetikzlibrary{arrows}
\usepackage[margin=1in]{geometry}
\newtheorem{theorem}{Theorem}[section]
\newtheorem{lemma}[theorem]{Lemma}
\newtheorem{definition}[theorem]{Definition}

\newtheorem{conjecture}[theorem]{Conjecture}

\newtheorem{fact}[theorem]{Fact}

\newtheorem{example}[theorem]{Example}

\newcommand{\ov}{\overline}

\DeclareMathOperator{\inn}{in}
\DeclareMathOperator{\out}{out}
\renewcommand{\bar}{\overline}

\DeclareMathOperator*{\E}{\mathbb{E}}

\DeclareMathOperator{\OPT}{OPT}

\makeatletter
\newcommand*{\RN}[1]{\expandafter\@slowromancap\romannumeral #1@}
\makeatother

\usepackage{lineno}

\usepackage{etoolbox}

\makeatletter

\newcommand{\define}[4][ignore]{%
  \ifstrequal{#1}{ignore}{}{
  \@namedef{thmtitle@#2}{#1}}%
  \@namedef{thm@#2}{#4}%
  \@namedef{thmtypen@#2}{lemma}%
  \newtheorem{thmtype@#2}[theorem]{#3}%
  \newtheorem*{thmtypealt@#2}{#3~\ref{#2}}%
}

\newcommand{\state}[1]{%
  \@namedef{curthm}{#1}
  \@ifundefined{thmtitle@#1}{
  \begin{thmtype@#1}
    }{
  \begin{thmtype@#1}[\@nameuse{thmtitle@#1}]
  }
    \label{#1}
    \@nameuse{thm@#1}
  \end{thmtype@#1}
  \@ifundefined{thmdone@#1}{
  \@namedef{thmdone@#1}{stated}%
  }{}
}

\newcommand{\restate}[1]{%
  \@namedef{curthm}{#1}
  \@ifundefined{thmtitle@#1}{
    \begin{thmtypealt@#1}
    }{
  \begin{thmtypealt@#1}[\@nameuse{thmtitle@#1}]
  }
    \@nameuse{thm@#1}
  \end{thmtypealt@#1}
  \@ifundefined{thmdone@#1}{
  \@namedef{thmdone@#1}{stated}%
  }{}
}

\newcommand{\thmlabel}[1]{
  \@ifundefined{thmdone@\@nameuse{curthm}}{\label{#1}
    }{\tag*{\eqref{#1}}}
}
\makeatother

\begin{document}

\runningauthor{Liau, Price, Song, Yang}
\twocolumn[

\aistatstitle{Stochastic Multi-armed Bandits in Constant Space}

\aistatsauthor{
  David Liau\\
  \texttt{davidliau@utexas.edu}\\
  \And
  Eric Price \\
  \texttt{ecprice@cs.utexas.edu}\\
  \And
  Zhao Song \\
  \texttt{zhaos@utexas.edu}\\
  \And
  Ger Yang \\
  \texttt{geryang@utexas.edu}\\
  
}

\aistatsaddress{The University of Texas at Austin}

]
  \begin{abstract}
We consider the stochastic bandit problem in the sublinear space
setting, where one cannot record the win-loss record for all $K$ arms.
We give an algorithm using $O(1)$ words of space with regret
\[
 \sum_{i=1}^{K}\frac{1}{\Delta_i}\log \frac{\Delta_i}{\Delta}\log T
\] where $\Delta_i$ is the gap between the best arm and arm $i$ and
$\Delta$ is the gap between the best and the second-best arms.  If the
rewards are bounded away from $0$ and $1$, this is within an $O(\log
1/\Delta)$ factor of the optimum regret possible without space
constraints.


  \end{abstract}



\section{Introduction}
In this paper, we study the multi-arm bandit problem in a sublinear
space setting.  In an instance of the bandit problem, there are $K$
arms and a finite time horizon $1,\dots,T$, where $T$ could be unknown
to us.  At each time step, we pull one of the $K$ arms, and receive a
reward that depends on our choice.  The goal is to find a strategy
that would achieve a sublinear (with respect to time) \emph{regret},
which is defined as the difference between the cumulative reward we
received from our strategy and the reward we could have received if we
always pulled the best arm in the hindsight.

There are many formulations of the bandit problem.  In this paper we
consider the stochastic setting specifically.  In
the stochastic setting, one assumes the rewards from the $i$-th arm
are i.i.d. random variables, with mean $\mu_i$ and support $[0,1]$.  A
well-known algorithm for the stochastic bandit is the UCB algorithm
\citep{ACF02}, and it is known that UCB achieves regret $O(K \log T)$.

The UCB algorithm requires $\Omega(K)$ space since it records the estimated rewards from all of the $K$ arms. However, in settings with limited space such as streaming algorithms, or settings with infinitely many arms \citep{K04}, the requirement is problematic.  There is a significant literature addressing this problem, but existing approaches assume structural properties on the set of arms,
e.g. combinatorial structure \citep{CL12} or continuum arm with local Lipschitz condition \citep{K04}.  A natural question is, what can we do without these structural assumptions given limited space?

A particular example is in a streaming algorithm setting, where space is much more limited than time, such as a router \citep{z13}. If the space constraint is $o(K)$ but the time constraint is $\Omega(K)$, one cannot run traditional UCB. In this case, $O(K)$ regret is still acceptable, and by accepting O(K) total regret, we can avoid requiring structural assumptions.  In a router, complicated strategy would corresponds to a larger set $K$ of possible strategies, which grants us the tradeoff: larger $K$ will result in a higher regret with a better optimum. 
Since routers have strict space constraints, running UCB would result in an extremely small regret on average over time ($K/T=\text{space}/\text{time}$, which is acceptable for routers). Our algorithm provides more flexibility in this bias/variance tradeoff.

\paragraph{Our techniques.}
Our algorithm is based on fairly simple ideas.  First, suppose we know
the time horizon $T$ and the expected value of the optimal arm $\mu^*$.  We
could then make a single pass through the arms; for each arm $i$, flip
it until we have high ($1 - 1/T^3$) confidence that
$\Delta_i = \mu^* - \mu_i > 0$, where $\mu_i$ is the expected value of arm $i$.  Once this happens, move to the next
arm.  This will flip each arm $O(\frac{\log T}{\Delta_i^2})$ times,
inducing regret $O(\frac{\log T}{\Delta_i^2} \cdot \Delta_i)$ from
this arm.  The total regret will then be
$O(\sum_{i \neq i^*} \frac{\log T}{\Delta_i})$, which is ideal, with
only constant space required.  The problem is that we don't know $T$
or $\mu^*$.  Not knowing $T$ isn't a big deal -- we can partition the
time horizon into $\log \log T$ scales, and the last $\log T$ term
will dominate \citep{AO10} -- but not knowing $\mu^*$ is a serious
problem.

We solve this problem by iteratively refining upper and lower bounds
$\mu_{LB}$ and $\mu_{UB}$ on $\mu^*$.  In each pass through the data,
we get new estimates that are half as far from each other.  After
$O(\log(1/\Delta))$ passes, where $\Delta=\min_{i:\Delta_i>0} \Delta_i$ is the minimal gap between the optimal and the suboptimal arms, only the best arm $i^*$ will remain
in the interval.  This gives an algorithm that loses at most an
$O(\log(1/\Delta))$ factor in the regret.  In some cases, the                                                     
loss is significantly smaller. Therefore, we can obtain the following result that improves the $O(\log(1/\Delta))$ factor into a $O(\log(\Delta_i/\Delta))$ factor,
\begin{theorem}
	\label{thm:main_informal}
Given a stochastic bandit instance with $K$ arms and their expected values $\mu_{1}, \cdots \mu_k \in [0,1]$. Let $\mu_* = \max_{i\in [K]} \mu_i$, $\Delta_i = \mu_* - \mu_i$, and $\Delta = \min_{i : \Delta_i > 0} \Delta_i$. For any $T>0$, there exists an algorithm that uses $O(1)$ words of space and achieves regret
\begin{align*}            
O\left( \sum_{i : \Delta_i>0} \frac{1}{\Delta_i} \log  \frac{\Delta_i}{\Delta}  \log T \right).
\end{align*}
\end{theorem}
Recall that the well-known UCB algorithm gives regret $O(\sum_{i:\Delta_i>0} \frac{\log T}{\Delta_i})$.
Our algorithm is always within a $\log (\Delta_i/\Delta)$ factor of its space-unlimited version.  In certain situations, we can do slightly better by refining our estimate of $\mu^*$ by more than a constant factor in each iteration.  This gives us the following result
\begin{theorem}\label{thm:main_informal2}
Under the same setting as Theorem~\ref{thm:main_informal}, for any  $\gamma>0$, there exists an algorithm that uses $O(1)$ words of space and achieves regret 
\begin{align*}
	O \left( \sum_{i:\Delta_i>0}  \frac{1}{\Delta_i} \left(\log^{\gamma}\frac{1}{\Delta_i} + \frac{\log (\Delta_i/\Delta)}{\gamma \log \log (\Delta_i/\Delta)} \right)\log(T) \right).
	\end{align*}
\end{theorem}
In particular, if we set $\gamma = 1/2$, we can find that this algorithm is always within an $O\left(\frac{\log(1/\Delta)}{\log \log(1/\Delta)}\right)$ factor of the space-unlimited UCB algorithm.

The paper is presented in the following manner. Section~\ref{sec:related}
reviews the related work.
 Section~\ref{sec:preli} provides
detailed preliminaries of problem formulation and the background
needed for our result. Section~\ref{sec:ucb} and \ref{sec:ucb3} contains the algorithm that gives the result (\RN{1}) and (\RN{2}) of Theorem~\ref{thm:main_informal} with known time horizon $T$, respectively.  Section~\ref{sec:unknown_T} demonstrates how to extend the algorithms to the case with unknown time horizon. The full version is available at \url{https://arxiv.org/pdf/1712.09007}.

\section{Related Works}\label{sec:related}

For stochastic bandits, the seminal work by \cite{LR85} demonstrated the idea of using the confidence intervals to solve the problem, and it showed that the lower bound of the regret is $\Omega(\sum \frac{\Delta_i \log T}{\mathrm{KL}(\mu_i,\mu_*)})$. The UCB algorithm, which is a simple solution to stochastic bandits, was analyzed in \cite{ACF02}. The UCB algorithm is based on Hoeffding's inequality, which is optimal when $\mathrm{KL}(\mu_i,\mu_*)\approx \Delta_i^2$.  In certain situations this can be improved using different types of concentration inequalities; for example, \cite{AMS09} used Bernstein's inequality to derive an algorithm with regret depending on the second moments.  Later, \cite{GC11} and \cite{MMG11} independently proposed the KL-UCB algorithm that matches the lower bound. We refer to the reader the comprehensive survey by \cite{BC12} for general bandit problems.

In addition to regret analysis for online decision making, there is a set of papers that discuss the sample complexity for the pure exploration problem, i.e. how to identify the best arm  \citep{MT04,EMM02,JMN14,KKS13,KCG15,EMM06}.  
Similar algorithms has been used in the regime of online decision making \citep{BJM11, AO10}.  With the idea of the best arm identification, the explore-then-commit (ETC) policy is designed to first performs some tests to identify the best arm, and then commit to it in the remaining time horizon.  The ETC policy is shown to be suboptimal \citep{garivier2016explore} but simplifies the analysis.   In particular, our algorithm is based on the framework by \cite{AO10}, but our algorithm takes only $O(1)$ space while the method by \cite{AO10} takes $O(K)$ space.


Moreover, there is a small set of papers that integrates the sketching techniques from streaming and online learning \citep{HS09,LAC16}.   \cite{HS09} considered the problem of minimizing $\alpha$-exp-concave losses, and the regret is required to be $O(\log T)$ uniformly over time.  They used the idea from streaming to keep a small active set of experts.  \cite{LAC16} considered the online convex optimization problem, and they used the ideas of sketching to reduce the efficiency for computing online Newton steps, however, the complexity is still $\Omega(K)$. 
\section{Preliminary}\label{sec:preli}
\paragraph{Notations}
For any positive integer $n$, we use $[n]$ to denote the set $\{1,2,\cdots,n\}$. 
For random variable $X$, let $\mathbb{E}[X]$ denote its expectation of $X$ (If this quantity exists). 
In addition to $O(\cdot)$ notation, for two functions $f,g$, we use the shorthand $f\lesssim g$ (resp. $\gtrsim$) to indicate that $f\leq C g$ (resp. $\geq$) for an absolute constant $C$. We use $f\eqsim g$ to mean $cf\leq g\leq Cf$ for constants $c,C$.

We measure space in words using the word RAM model, so that the input values (such as K, T, and rewards) and variables can each be expressed in $O(1)$ word of space in $O(\log(KT))$ bits. For more details of word RAM model, we refer the readers to \cite{ahu74,clrs09}.

\subsection{Problem Formulations} \label{sec:problem_formulation}
\begin{definition}
For a \emph{multi-armed bandit problem}, there are $K$ arms in total, and a finite time horizon $1,2,\dots,T$.  At each time step $t \in [T]$, the player has to choose an arm $I_t \in [K]$ to play, and receives a reward $X_{i,t}$ associate to that arm.  Without loss of generality, assume that for each arm $i \in [K]$ and each time step $t \in [T]$, $X_{i,t} \in [0,1]$.  We denote the arm that player chooses at time $t$ as $I_t$. The goal of the player is to maximize the total reward he is getting.  We will measure the performance of an algorithm via its \emph{regret}, which is defined as the difference between the best reward in the hindsight and the reward received with the algorithm:
$$ \Psi_T = \max_{i \in [K]} \left( \sum_{t=1}^{T} X_{i,t} - \sum_{t=1}^{T} X_{I_t,t} \right). $$
\end{definition}

In this paper, we consider the stochastic setting, where we assume the rewards are coming from some stochastic processes.
\begin{definition}
In a \emph{stochastic bandit}, we assume each arm $i \in [K]$ is associated with a distribution ${\cal D}_i$ over $[0,1]$, with mean $\mu_i$.  The reward $X_{i,t}$ at time $t \in [T]$ is drawn from ${\cal D}_i$ independently.
\end{definition}

For stochastic bandits, instead of using the regret defined above, we will consider the \emph{pseudo regret}:
$$ \ov{\Psi}_T = \max_{i \in [K]} \left( \E  \left[ \sum_{t=1}^{T} X_{i,t} \right] - \E \left[ \sum_{t=1}^{T} X_{I_t,t} \right] \right) .$$
We can rewrite the pseudo regret using  Wald's identity:
\begin{equation} \label{eq:pregret_wald}
	\ov{\Psi}_T = \max_{i \in [K]} \left( \sum_{j=1}^{K} \E \left[ N_{j,T} \Delta_{ij} \right] \right),
\end{equation}
where $N_{j,T}$ is the number of times arm $j$ is chosen up to time $T$,
and we define $\Delta_{ij} = \mu_i - \mu_j$ to be the gap between the means of arm $i$ and arm $j$.  
We use $\mu_{*}$ to denote the mean reward for the arm with the highest mean, i.e., $\mu_* = \max_{i\in [K] } ~ \mu_i$.



\subsection{Concentration Inequalities}
In this paper, for simplicity, we will use Chernoff-Hoeffding inequality to analyze the concentration behavior for random variables with bounded support.
\begin{fact}[Chernoff-Hoeffding Bound]
	Let $x_1$, $x_2$, $\dots$, $x_n$ be i.i.d. random variables in $[0,1]$.  Let $ X = \frac{1}{n}\sum_{i=1}^{n} x_i$.  Then for any $\epsilon > 0$,
	\begin{align*}
		\Pr\left[ | X - \E [X] | > \epsilon \right] \le 2e^{-2n\epsilon^2}.
	\end{align*}
\end{fact}

\section{UCBConstSpace with known T}\label{sec:ucb}
The original UCB-1 algorithm \citep{ACF02} needs $O(K)$ space to achieve $O(\sum_{i:\Delta_i>0} \frac{1}{\Delta_i} \log T )$ regret. In this section, we propose a new algorithm which requires only $O(1)$ space in exchange for a slightly worse regret. 

First, we consider the setting where $T$ is known.  The main result is presented in the following theorem.

\begin{theorem}\label{thm:ucb}
Given a stochastic bandit instance with known $T$, let $\Delta_i = \mu_* - \mu_i$, and let $\Delta = \min_{i:\Delta_i>0} \Delta_i$. Then for any $T>0$, there exists an algorithm that uses $O(1)$ words of space and achieves regret
\begin{align*}
O \left( \sum_{i:\Delta_i>0} \frac{1}{\Delta_i}  \log (\Delta_i/\Delta ) \log T \right).
\end{align*}
\end{theorem}

\begin{algorithm*}[t]\caption{UCB algorithm with constant space and known $T$ (Theorem~\ref{thm:ucb})}
	\label{alg:ucb_known_T}
	\begin{algorithmic}[1]{
		\Procedure{\textsc{UCBConstSpace}}{$K,T$}
		\State Set $\delta \gets 1/T^3$, initialize $g_1 \leftarrow \frac{1}{2}$, $t\leftarrow 1$
		\State {\bf Exploration Phase:}
		\For{rounds $r=1,2,\dots$}
			\State $a'$: the best arm in the previous round, $\bar{\mu}'$: mean reward for arm $a'$ in the previous round
			\State $N \leftarrow \lceil 2  \log(1/\delta)  / g_r^2 \rceil $, which is the maximum number of plays for each arm in the current round
			\State Initialize $a,b\gets 0$, which are the best and the second best arm in this round
			\State Initialize $\bar \mu_a, \bar \mu_b \gets 0$, which are the means for arms $a$ and $b$
			\For{each arm $i=1\to K$}
				\State Set $\bar \mu \gets 0$, which keeps the mean reward for arm $i$ in the current round
				\For{$n=1\to N$}
					\State Pull arm $i$ and receive reward $v$
					\State $t \gets t+1$
					\State Update $\bar \mu$ with $v$: $\ov{\mu} \leftarrow ( \ov{\mu} \cdot (n-1) + v ) /n$
					\If {$\ov{\mu} + \sqrt{\log(1/\delta)/2n}  < \ov{\mu}' - g_{r-1}/2 $}
						\State {\bf break}, i.e. we \emph{rule out} arm $i$ for the current round
					\EndIf
				\EndFor
				\State {\bf if} $\bar \mu> \bar \mu_a$ {\bf then} $b \gets a$, $\bar \mu_b \gets \bar \mu_a$, $a \gets i$ and $\bar \mu_a \gets \bar \mu$ \Comment{Update the best and the 2nd best arms}
				\State {\bf else if} $\bar \mu> \bar \mu_b$ {\bf then} $b \gets i$ and $\bar \mu_b \gets \bar \mu$ \Comment{Update the 2nd best arm}
			\EndFor
			\State {\bf Stopping Criterion: if} $\bar \mu_a - g_r/2 > \bar \mu_b + g_r/2$ or $t>T$ {\bf then break} \label{algl:loop_until}
			\State Update $a' = a$ and $\bar \mu' = \bar \mu_{a}$
			\State {\bf Set new precision:} $g_{r+1} = g_{r}/2$ \label{algl:update_Delta}
		\EndFor
		\State {\bf Exploitation Phase:}
		\State Pull arm $a$ for the remaining time steps.
		\EndProcedure}
	\end{algorithmic}
\end{algorithm*}

We present the method in Algorithm~\ref{alg:ucb_known_T}, where we iteratively improve our estimation of $\Delta$.  More precisely, we scan through the data multiple rounds.  In the $r$-th round, we sample each arm up to some precision $g_r$.  The desired precision $g_r$ is halved after each round.  In this sampling process, we only keep the information of the best arm and the second best arm seen in the current and the previous round, instead of saving those from all arms.  With the information of the best arm and the current precision $g_r$, we can refine the upper and lower bound $\mu_{UB}$ and $\mu_{LB}$ on $\mu_*$.  If an arm whose upper confidence value is less than $\mu_{LB}$, we can rule it out without continuing to $g_r$ precision.  This process is terminated if we are able to determine the best arm with the rest arms.

We define $a^{(r)}$ and $b^{(r)}$ as the best arm and the second best arm stored at the end of the $r$-th round.  Also, we let $\bar{\mu}_i^{(r)}$ to be the recorded empirical mean at the end of the $r$-th round for arm $i$.  Denote $n_{i}^{(r)}$ as the total number of pulls of arm $i$ at the $r$-th round. Then, we define $\bar{\mu}_{i,n}^{(r)}$ as the empirical mean $\bar{\mu}_i$ stored for arm $i$ after pulling it for $n$ times in round $r$.  Further, we define $r_{\max}$ as the value of $r-1$ at the moment the algorithm exits the loop in Line~\ref{algl:loop_until}. 

\begin{definition}
For each $r \in [r_{\max}]$, define the event $\xi_r$ to be the event:
$
 \exists r' \in [r], \exists i \in [K], \exists n \in [n_i^{(r')}] 
$
such that
$
 |\bar{\mu}_{i,n}^{(r')} - \mu_i| > \sqrt{ \log(1/\delta) /(2n)} 
$,
i.e., there exists some estimate of $\overline{\mu}_{i,n}^{(r')}$ that is not within our desired confidence interval up to round $r$. 
\end{definition}
 Throughout the first part of our analysis, we focus on the case when $\neg \xi_r$ holds when we are discussing the state of the algorithm at round $r$, i.e., all estimates are within our desired confidence interval.
\begin{lemma} \label{lem:ucb_error_round}
	In Algorithm~\ref{alg:ucb_known_T}, at any round $r \in [r_{\max}]$, given $\neg \xi_r$, the following statements are true: \\
		1. $n_{a^{(r)}}^{(r)} = \lceil ( 2\log(1/\delta) ) / g_r^2  \rceil$, i.e. the claimed optimal arm cannot be ruled out early. \\
		2. $n_*^{(r)} = \lceil  (2\log(1/\delta) )/ g_r^2 \rceil$, i.e. the true optimal arm cannot be ruled out early. \\
		3. $|\bar{\mu}_{a^{(r)}}^{(r)} - \mu_*| \le g_r/2$.
\end{lemma}
\begin{proof}
	We prove this lemma by induction.  For the base case, the first and the second statement are true because all arms have to be played for $\lceil \frac{2\log(1/\delta)}{(g_1)^2} \rceil$ times.  For the third statement, we prove by contradiction.  Assume the contrary, i.e. $\bar{\mu}_{a^{(1)}}^{(1)} - \mu_* > g_1/2$ or $\mu_* - \bar{\mu}_{a^{(1)}}^{(1)} > g_1/2$.  If $\bar{\mu}_{a^{(1)}}^{(1)} - \mu_* > g_1/2$, then we have
	\begin{align*}
	\mu_* &< ~ \bar{\mu}_{a^{(1)}}^{(1)} - g_1/2 &\\
	&\le ~\mu_{a^{(1)}} + \sqrt{ \log(1/\delta) / (2 n_{a^{(1)}}^{(1)}} )  - g_1/2 \\
	&\le~ \mu_{a^{(1)}} + g_1/2 - g_1/2 \\
	&= ~\mu_{a^{(1)}}
	\end{align*}
	where the second step follows by condition $\neg\xi_r$ and the third step follows by $n_{a^{(1)}}^{(1)}\geq   ( 2\log(1/\delta) ) /g_1^2 $.
	
	The above equation leads to a contradiction because $\mu_*>\mu_{i}$ for any $i\not=*$.  Similarly, if $\mu_* - \bar{\mu}_{a^{(1)}}^{(1)} > g_1/2$, then we have
	\begin{align*}
	\bar{\mu}_{a^{(1)}}^{(1)}  &< \mu_* - g_1/2 &\\
	&\le \bar{\mu}_{*}^{(1)} + \sqrt{ \log(1/\delta) /( 2 n_{*}^{(1)}}) - g_1/2 \\
	&\le \bar{\mu}_{*}^{(1)} + g_1/2 - g_1/2 \\
	&= \bar{\mu}_{*}^{(1)}
	\end{align*}
	where the second step follows by condition $\neg \xi_r$, and the third step follows by $n_{*}^{(1)}\ge (2\log(1/\delta) )/g_1^2 $.
	
	The above equation also results in a contradiction because for any $i\not=*$ to be assigned as $a^{(1)}$, we must have $\bar{\mu}_{a^{(1)}}^{(1)} > \bar{\mu}_{*}^{(1)}$.
	
	For the induction step, we assume these three statements are true for $r \le r'-1$.  Now consider $r=r'$.  We first prove the second statement.  Assume the contrary, i.e. the true optimal arm has been ruled out early, meaning
	\begin{equation} \label{eq:opt_arm_ruled_out_early_cond}
	\bar{\mu}_{*}^{(r)} + \sqrt{ \log(1/\delta) / ( 2 n_{*}^{(r)}} ) < \bar{\mu}_{a^{(r-1)}}^{(r-1)} - g_{r-1}/2
	\end{equation}
	Then, we can see that
	\begin{align}
	\mu_* &\le ~ \bar{\mu}_{*}^{(r)} + \sqrt{\log(1/\delta) / (2 n_{*}^{(r)}}) 
	\nonumber\\
	&< ~ \bar{\mu}_{a^{(r-1)}}^{(r-1)} - g_{r-1}/2
	\nonumber\\
	&\le ~\mu_{a^{(r-1)}} \label{eq:lem1_contradict1}
	\end{align}
	where in the last inequality, we use the induction hypothesis, $n_{a^{(r-1)}}^{(r-1)} \ge \frac{2 \log(1/\delta)}{g_{r-1}^2}$ and then
	\begin{align*}
	\mu_{a^{(r-1)}} \ge \bar{\mu}_{a^{(r-1)}}^{(r-1)} - \sqrt{\frac{\log(1/\delta)}{2n_{a^{(r-1)}}^{(r-1)}}} \ge \bar{\mu}_{a^{(r-1)}}^{(r-1)} - g_{r-1}/2
	\end{align*}
	There is a contradiction in \eqref{eq:lem1_contradict1} because we must have $\mu_* \ge \mu_{a^(r-1)}$.  Hence the second statement is true.
	
	Next, we can see that the first statement is now clear because we have shown that there is at least one arm that is going to pull for $\lceil \frac{2\log(1/\delta)}{g_r^2} \rceil$ times at the $r$-th round (which is arm $*$ according to the second statement we have just shown).  This means that if arm $a^{(r)}$ is not arm $*$, then it has to be pulled for $\lceil \frac{2\log(1/\delta)}{g_r^2} \rceil$ times as well.

	For the third statement, the proof is similar to the base case, where we prove by contradiction.  Assume the contrary, i.e. $\bar{\mu}_{a^{(r)}}^{(r)} - \mu_* > g_r/2$ or $\mu_* - \bar{\mu}_{a^{(r)}}^{(r)} > g_r/2$.
	
	If $\bar{\mu}_{a^{(r)}}^{(r)} - \mu_* > g_r/2$, then we have
	\begin{align*}
	\mu_* &< \bar{\mu}_{a^{(r)}}^{(r)} - g_r/2 &\\
	&\le \mu_{a^{(r)}} + \sqrt{ \log(1/\delta) / ( 2 n_{a^{(r)}}^{(r)}}) - g_r/2 \\
	&\le \mu_{a^{(r)}} + g_r/2 - g_r/2 \\
	&= \mu_{a^{(r)}}
	\end{align*}
	where the second step follows by condition~$\neg \xi_r$, and the third step follows by $n_{a^{(r)}}^{(r)}\ge \frac{2\log(1/\delta)}{g_r^2} $(the first statement).
	
	This results in a contradiction because $\mu_*\ge\mu_{i}$ for any $i\in[K]$.  Similarly, if $\mu_* - \bar{\mu}_{a^{(r)}}^{(r)} > g_r/2$, then we have
	\begin{align*}
	\bar{\mu}_{a^{(r)}}^{(r)}  &< \mu_* - g_r/2 &\\
	&\le \bar{\mu}_{*}^{(r)} + \sqrt{ \log(1/\delta) / ( 2 n_{*}^{(r)}}) - g_r/2 \\
	&\le \bar{\mu}_{*}^{(r)} + g_r/2 - g_r/2 \\
	&= \bar{\mu}_{*}^{(r)}
	\end{align*}
	where the second step follows by condition~$\neg \xi_r$ and the third step follows by $n_{*}^{(r)}\ge \frac{2\log(1/\delta)}{g_r^2} $(the second statement). 
	
	This results in a contradiction because for any $i\not=*$ to be assigned as $a^{(r)}$, we must have $\bar{\mu}_{a^{(r)}}^{(r)} > \bar{\mu}_{*}^{(r)}$, otherwise we will have $|\bar{\mu}_{*}^{(r)} - \mu_*| \le g_r/2$ by condition $\neg \xi_r$.
\end{proof}

\begin{lemma} \label{lem:ucb_max_R}
	In Algorithm~\ref{alg:ucb_known_T}, conditioning on event $\neg \xi_{r_{\max}}$ holds, we have $r_{\max} \le \lceil \log(2/\Delta) \rceil$.
\end{lemma}
\begin{proof}
	Assume the contrary, i.e. at the end of round $r=\lceil \log(2/\Delta) \rceil$, the best arm and the second best arm are still not differentiated, meaning we still have
	\begin{align*} 
	\bar{\mu}_*^{(r)} - g_r/2 &< \bar{\mu}_{a^{(r)}}^{(r)} + g_r/2
	\end{align*}
	First note that $r > \log(2/\Delta)$ implies $2^{-r}=g_r < \Delta/2$.  We have
	\begin{align*}
		\mu_* &\le \bar{\mu}_*^{(r)} + \sqrt{ \log(1/\delta) / ( 2n_*^{(r)}}) \\
		&\le \bar{\mu}_*^{(r)} + g_r/2 < \bar{\mu}_{a^{(r)}}^{(r)} + 3g_r/2\\
		&< \bar{\mu}_{a^{(r)}}^{(r)} + 3\Delta/4
	\end{align*}
	Similarly, we have
	$
		\mu_{a^{(r)}} > \bar{\mu}_{a^{(r)}}^{(r)} - \Delta/4 \label{eq:mu_aR_R_lb}
	$.
	Then, we can show that
	\begin{align*}
		\Delta \le \mu_* - \mu_a \le (\bar{\mu}_{a^{(r)}}^{(r)} + 3\Delta/4) - (\bar{\mu}_{a^{(r)}}^{(r)} - \Delta/4) < \Delta 
	\end{align*} 
	which results in a contradiction.  This implies that given $\neg \xi_{r_{\max}}$, we must have $r_{\max} \le \lceil \log(2/\Delta) \rceil$.
\end{proof}

\begin{lemma} \label{lem:ruled_out1}
	In Algorithm~\ref{alg:ucb_known_T}, at any round $r$, given $\neg \xi_{r}$, the number of plays for any arm $i\in[K]$ is upper-bounded by
	$$ n_i^{(r)} \le  \frac{2\log(1/\delta)}{\left(\Delta_i- g_{r-1} \right)^2} +1.$$
\end{lemma}
\begin{proof}
	First, note that as long as an arm has not been ruled out, we have
	\begin{equation} \label{eq:before_ruleout3}
	\bar{\mu}_{i,n_i^{(r)}-1}^{r} + \sqrt{\frac{\log(1/\delta)}{2(n_i^{(r)}-1)}} \ge \bar{\mu}_{a^{(r-1)}}^{(r-1)} - \frac{g_{r-1}}{2}
	\end{equation}
	Then, we can show
	\begin{align}
	\Delta_i &= \mu_* - \mu_i  \nonumber\\
	&\le \bar{\mu}_{a^{(r-1)}}^{(r-1)} + \frac{g_{r-1}}{2} - \mu_i  \nonumber\\
	&\le \bar{\mu}_{a^{(r-1)}}^{(r-1)} + \frac{g_{r-1}}{2} - \left( \bar{\mu}_{i,n_i^{(r)}-1}^{(r)}-\sqrt{\frac{\log(1/\delta)}{2(n_i^{(r)}-1)}} \right) \notag\\
	&\le  2\left( \frac{g_{r-1}}{2} + \sqrt{\frac{\log(1/\delta)}{2(n_i^{(r)}-1)}} \right) \notag
	\end{align}
	where the second step follows from Lemma~\ref{lem:ucb_error_round}, the third step follows by $\neg \xi_r$, and the last step follows by \eqref{eq:before_ruleout3}.  Reorganizing the above inequality proves the lemma.
\end{proof}

\begin{proof}[Proof of Theorem~\ref{thm:ucb}]
	Consider Algorithm~\ref{alg:ucb1_unknown_T}.  For each round $r \in [r_{\max}]$, conditioned on $\neg \xi_r$, i.e. the confidence interval is correct, we first recognize two bounds on the number of plays $n_i^{(r)}$ for each arm $i \in [K]$.
	
	By the definition of Algorithm~\ref{alg:ucb_known_T}, we have
	\begin{equation} \label{eq:n_case2}
	n_i^{(r)} \le \ 2\log(1/\delta) /  g_r^2 +1
	\end{equation}
	Also, from Lemma~\ref{lem:ruled_out1}, we have
	\begin{equation} \label{eq:n_case1}
	n_i^{(r)} \le  2\log(1/\delta) / (\Delta_i- g_{r-1} )^2 +1 
	\end{equation}
	By combining \eqref{eq:n_case2} and \eqref{eq:n_case1}, together with $r_{\max} \le \lceil \log(2/\Delta) \rceil$ by Lemma~\ref{lem:ucb_max_R}, we can upper bound the regret results from pulling arm $i$ in the algorithm. Let $\alpha=\lceil \log(2/\Delta) \rceil$ and $\beta=\lceil\log(3/\Delta_i)\rceil$. Conditioning on event $\neg \xi_{r_{\max}}$ holds, we have,
	\begin{align*}
		 \sum_{r=1}^{\alpha} \Delta_i n_{i}^{(r)} \le &~ \sum_{r=1}^{\alpha} \Delta_i \left( \frac{2\log(1/\delta)}{\left( \max\left\{g_r, \Delta_i-2g_r \right\} \right)^2} +1 \right) &\nonumber\\
		= & ~\sum_{r=1}^{\alpha} \Delta_i \left( \frac{2\log(1/\delta)}{\left( \max\left\{2^{-r}, \Delta_i-2 \cdot 2^{-r} \right\} \right)^2} +1 \right) &\nonumber
	\end{align*}
	Furthermore, we can obtain {\footnotesize
	\begin{align}
		&\quad \sum_{r=1}^{\alpha} \Delta_i n_{i}^{(r)} \notag \\
		&\le \sum_{r=1}^{\beta} \Delta_i \cdot \frac{2\log(1/\delta)}{2^{-2r}} \notag + \sum_{r=\beta+1}^{\alpha} \Delta_i \cdot \frac{2\log(1/\delta)}{\left(\Delta_i - 2 \cdot 2^{-\log(3/\Delta_i)}\right)^2} \notag \\
		& + \Delta_i \cdot \lceil \log(2/\Delta) \rceil &\nonumber\\
		&\le \frac{288\log(1/\delta)}{\Delta_i} + \frac{18\log(2\Delta_i/3\Delta)\log(1/\delta)}{\Delta_i} \notag\\
		& + \Delta_i (\log(2/\Delta)+1) &\nonumber\\
		&\lesssim \frac{\log(\Delta_i/\Delta)\log(1/\delta)}{\Delta_i} \label{eq:regret_per_arm}
	\end{align}}
	For the next step, we find an upper bound for the probability of event $\xi_{r_{\max}} :=\big\{\exists r\in[r_{\max}] ,\exists i \in [K], \exists n \in [n_i^{(r)}] \text{ s.t. } |\bar{\mu}_{i,n}^{(r)} - \mu_i| > \sqrt{ \log(1/\delta) / (2n) }\big\}$: {\footnotesize 
	\begin{align}
		 & ~\Pr(\xi_{r_{\max}}) \notag\\
		 \le & ~\sum_{r=1}^{T/K} \sum_{i=1}^{K} \sum_{n=1}^{T} \Pr\left( |\bar{\mu}_{i,n}^{(r)} - \mu_i| > \sqrt{\log(1/\delta)/ (2n)}\right) \nonumber \\
		\le & ~ 2T^2 \delta 
		\label{eq:prob_E} 
	\end{align}}	
	Finally, by choosing $\delta = 1/T^3$, and combining \eqref{eq:regret_per_arm} and \eqref{eq:prob_E}, we have
	\begin{align*}
	 \ov{\Psi}_T & \lesssim ~ \sum_{i=1}^{K} \left(\frac{\log(\Delta_i/\Delta)\log(T)}{\Delta_i} + \Delta_i T \cdot 2T^2 \delta \right)  \\
	 & \lesssim ~ \sum_{i=1}^{K} \frac{\log(\Delta_i/\Delta)\log(T)}{\Delta_i} 
	 \end{align*}
	which proves the theorem.
\end{proof}

\section{Improved Algorithm for UCBConstSpace}\label{sec:ucb3}

The result in Theorem~\ref{alg:ucb1_unknown_T} gives an additional $O(\log(\Delta_i/\Delta))$ factor to the original UCB-1 algorithm by \cite{ACF02}.  This means that in a bad scenario, for example, if most of the arms have gap $\Delta_i = K \Delta$, the $O(\log(\Delta_i/\Delta))$ factor translates to an additional $\log K$ factor in the regret.

In this section, we show that we are able to improve the additional $\log(\Delta_i/\Delta)$ factor to a $\frac{\log(\Delta_i/\Delta)}{\log \log(\Delta_i/\Delta)}$ factor by slightly changing the update rule on the precision $g_r$.  This means that in the bad example described above, we are improving the competitive ratio from $\log K$ to $\frac{\log K}{\log \log K}$.  We present our result in the following theorem.
\begin{theorem}\label{thm:ucb3}
	Given a stochastic bandit instance with known $T$, let $\Delta_i = \mu_* - \mu_i$, and let $\Delta = \min_{i:\Delta_i>0} \Delta_i$. For any $\gamma>0$ and any $T>0$, there exists an algorithm that uses $O(1)$ words of space and achieves regret
	\begin{align*}
	O \left( \sum_{i:\Delta_i>0}  \frac{1}{\Delta_i} \left(\log^{\gamma}\frac{1}{\Delta_i} + \frac{\log (\Delta_i/\Delta)}{\gamma \log \log (\Delta_i/\Delta)} \right)\log(T) \right).
	\end{align*}
\end{theorem}

We consider a modified version of Algorithm~\ref{alg:ucb_known_T}, where the update rule in Line~\ref{algl:update_Delta} is replaced by 
\begin{equation}\label{eq:new_update_rule}
g_{r+1}=\frac{g_r}{2\left(\log(1/g_r)\right)^{\epsilon}}
\end{equation}
where $\epsilon$ is some constant to be determined later.
In the following lemma, we show that with this update rule, basically given any $D<1$, it takes only $O\left(\frac{1}{\epsilon}\cdot \frac{\log(1/D)}{\log \log(1/D)}\right)$ steps to reach accuracy $D$.
\begin{lemma} \label{lem:log_loglog}
	Given any $g_0, D \in (0,1)$, $D<g_0$, let $r_0 = \frac{\log (g_0/D)}{\log \log (g_0/D)}$.  If for any positive integer $r$,
	$g_{r}=\frac{g_{r-1}}{2\left(\log(1/g_{r-1})\right)^{\epsilon}}$.
	Then, for any $r \ge (\frac{2}{\epsilon}+1)r_0+2$, we have $g_r \le D$.
\end{lemma}
\begin{proof}
	First, note that by definition of $g_r$, we have $g_r \le g_0 2^{-r}$ for any $r \ge 1$.  Therefore, for any $r \ge r_0$, we have
	$$
	g_r \le g_0 2^{-r} \le g_0 2^{-r_0} = g_0 (D/g_0)^{\frac{1}{\log \log (g_0/D)}}
	$$
	Then, we can see that for any $r \ge r_0$, we have
	$$
	g_{r+1} \le \frac{g_r}{2\left(\frac{\log (g_0/D)}{\log \log (g_0/D)}\right)^{\epsilon}} \le \frac{g_r}{2(\log (g_0/D))^{\epsilon/2}}
	$$
	As a result, we have
	\begin{align*}
	g_{r+\lceil\frac{2}{\epsilon}r_0\rceil} &\le \frac{g_r}{2^{\frac{2}{\epsilon}r_0} (\log (g_0/D))^{r_0}}\\
	&\le g_0 (\log (g_0/D))^{-r_0} = D
	\end{align*}
	This implies that for any $r \ge (r_0 + 1) + (\frac{2}{\epsilon}r_0 + 1) \ge \lceil{r_0}\rceil + \lceil{\frac{2}{\epsilon}r_0}\rceil$, we have $g_r \le D$.
\end{proof}
Note that we can apply Lemma~\ref{lem:ucb_error_round} and Lemma~\ref{lem:ruled_out1} for Algorithm~\ref{alg:ucb1_unknown_T} with update rule \eqref{eq:new_update_rule} because they do not require specific update rules.  Before we proceed to the proof of Theorem~\ref{thm:ucb3}, we need the following lemma for an upper bound of $r_{\max}$.
\begin{lemma}\label{lem:r_max}
	In Algorithm~\ref{alg:ucb1_unknown_T} with update rule \eqref{eq:new_update_rule}, given $\neg \xi_{r_{\max}}$, we have $r_{\max} \le \lceil(\frac{2}{\epsilon}+1)\frac{\log 2/\Delta}{\log \log 2/\Delta} + 2\rceil$.
\end{lemma}
\begin{proof}
	Assume the contrary, i.e. at the end of round $r=\lceil(\frac{2}{\epsilon}+1)\frac{\log 2/\Delta}{\log \log 2/\Delta} + 2\rceil$, the best arm and the second best arm are still not differentiated, meaning for some $i \not= *$, we still have
	$$
	\bar{\mu}_*^{(r)} - g_r/2 < \bar{\mu}_{i}^{(r)} + g_r/2
	$$
	By Lemma~\ref{lem:log_loglog}, we have $g_r \le \Delta/2$.  Thus, we have
	\begin{align*}
		\mu_* &\le \bar{\mu}_*^{(r)} + g_r/2 \le \bar{\mu}_i^{(r)} + 3g_r/2 \le \bar{\mu}_i^{(r)} + 3\Delta/4
	\end{align*}
	where for the second step we use Lemma~\ref{lem:ucb_error_round}.  Similarly, we have$ \mu_i \ge \bar{\mu}_i^{(r)} - \Delta/4.$
	Then, we have
	\begin{align*}
		\Delta \le \mu_* - \mu_i \le (\bar{\mu}_i^{(r)} + 3\Delta/4) - (\bar{\mu}_{a^{(r)}}^{(r)} - \Delta/4) < \Delta 
	\end{align*} 
	which results in a contradiction.
\end{proof}
\begin{proof}[Proof of Theorem~\ref{thm:ucb3}]
	Consider Algorithm~\ref{alg:ucb1_unknown_T} with update rule \eqref{eq:new_update_rule}.  For each arm $i \in [K]$, if we condition on $\neg \xi_{r_{\max}}$, then by Lemma~\ref{lem:ruled_out1} and Lemma~\ref{lem:r_max}, we can upper bound the regret results from pulling arm $i$ in the algorithm:
	\begin{align}
	& ~\sum_{r=1}^{r_{\max}} \Delta_i n_{i}^{(r)} \notag \\
	\le & ~ \sum_{r=1}^{r_{\max}} \Delta_i \left( \frac{2\log(1/\delta)}{\left( \max\left\{g_r, \Delta_i-g_{r-1} \right\} \right)^2} +1 \right) &\nonumber\\
	\le& ~\sum_{r=1}^{r_i} \Delta_i \cdot \frac{2\log(1/\delta)}{g_r^2} + \sum_{r=r_i+1}^{r_{\max}} \Delta_i \cdot \frac{2\log(1/\delta)}{\left(\Delta_i - g_{r-1}\right)^2} \notag \\
	& ~ + \Delta_i \cdot r_{\max} \label{eq:ucb3_regret_derivation1}
	\end{align}
	where $r_i$ be the minimal round $r$ such that $g_r<\Delta_i/2$. For the first term of \eqref{eq:ucb3_regret_derivation1}, since $g_r$ decays super-exponentially, i.e. $g_{r+1} \le g_r/2$, we have
	\begin{align}
	\sum_{r=1}^{r_i} \frac{2\Delta_i\log(1/\delta)}{g_r^2} &\le \frac{4\Delta_i\log(1/\delta)}{g_{r_i}^2} \nonumber\\ &=\frac{4\Delta_i\left(\log\frac{1}{g_{r_i-1}}\right)^{2\epsilon}\log(1/\delta)}{g_{r_i-1}^2} \nonumber\\
	&\le \frac{16\left(\log\frac{1}{\Delta_i}\right)^{2\epsilon}}{\Delta_i}\log(1/\delta) \label{eq:ucb3_regret_derivation1_t1}
	\end{align}
	where the last step follows from the fact that $g_{r_i-1} \ge \Delta_i/2$ by the definition of $r_i$.  For the second term of \eqref{eq:ucb3_regret_derivation1}, we have
	\begin{align}
	 \sum_{r=r_i+1}^{r_{\max}} \frac{2\Delta_i\log(1/\delta)}{\left(\Delta_i - g_{r-1}\right)^2} 
	\le & ~\sum_{r=r_i+1}^{r_{\max}} \frac{8\Delta_i\log(1/\delta)}{\Delta_i^2} \nonumber \\
	\le & ~\sum_{r=r_i+1}^{r_{\max}} \frac{8\log(1/\delta)}{\Delta_i} \label{eq:ucb3_regret_derivation1_t2}
	\end{align}
	By Lemma~\ref{lem:log_loglog}, we can find that it takes $ \lceil(\frac{2}{\epsilon}+1)\frac{\log (\Delta_i/\Delta)}{\log \log (\Delta_i/\Delta)} + 2\rceil$ rounds to get from $\Delta_i/2$ to $\Delta_i/2$.  As a result, we can upper bound \eqref{eq:ucb3_regret_derivation1_t2} by
	\begin{align}
	&\sum_{r=r_i+1}^{r_{\max}} \frac{2\Delta_i\log(1/\delta)}{\left(\Delta_i - g_{r-1}\right)^2} \nonumber \\
	&\le \left(\left(\frac{2}{\epsilon}+1\right)\frac{\log (\Delta_i/\Delta)}{\log \log (\Delta_i/\Delta)} + 3\right) \frac{8\log(1/\delta)}{\Delta_i} \nonumber \\
	&\le \left(\frac{2}{\epsilon}+1\right)\frac{\log (\Delta_i/\Delta)}{\log \log (\Delta_i/\Delta)} \frac{16\log(1/\delta)}{\Delta_i} \label{eq:ucb3_regret_derivation1_t2p}
	\end{align}
	
	Using the similar argument as we have done in the proof of Theorem~\ref{thm:ucb}, we can find that 
	\begin{equation} \label{eq:prob_E2}
	\Pr(\xi_{r_{\max}}) \le 2T^2\delta
	\end{equation}
	
	Finally, by combining \eqref{eq:ucb3_regret_derivation1_t1}, \eqref{eq:ucb3_regret_derivation1_t2p}, and \eqref{eq:prob_E2}, we can get {\footnotesize
	\begin{align*}
	 \ov{\Psi}_T \le & ~ 16\sum_{i:\Delta_i>0} \frac{1}{\Delta_i}\bigg(\log^{2\epsilon}\frac{1}{\Delta_i} + \left(\frac{2}{\epsilon}+1\right)\frac{\log (\Delta_i/\Delta)}{\log \log (\Delta_i/\Delta)}\bigg) \\
	& ~\cdot \log(1/\delta)  + \sum_{i:\Delta_i>0} \Delta_i T \cdot 2T^2 \delta
	\end{align*}}
	By choosing $\delta=1/T^3$ and $\epsilon = \gamma/2$ we can find that 
	\begin{align*}
	\ov{\Psi}_T \lesssim \sum_{i:\Delta_i>0}  \frac{1}{\Delta_i} \left(\log^{\gamma}\frac{1}{\Delta_i} + \frac{\log (\Delta_i/\Delta)}{\gamma \log \log (\Delta_i/\Delta)} \right)\log(T)
	\end{align*}
	which proves the theorem.
\end{proof}
We conjecture below that the $O(\frac{\log (\Delta_i/\Delta)}{ \log \log (\Delta_i/\Delta)})$ factor is not improvable given the $O(1)$ space constraint.  The discussion for our conjectured hard instance is in the appendix.
\begin{conjecture}\label{con:lower_bound}
	There exists a distribution over stochastic bandit problems such that, for any algorithm taking $O(1)$ words of space will have regret 
	\begin{align*}
	\Omega \left( \sum_{i:\Delta_i>0}  \frac{1}{\Delta_i} \left(\frac{\log (\Delta_i/\Delta)}{ \log \log (\Delta_i/\Delta)} \right)\log(T) \right).
	\end{align*}
\end{conjecture}

\section{Unknown Horizon $T$} \label{sec:unknown_T}
Now, we show that using the technique described in \citep{AO10}, we are able to get the same regret as in Theorem~\ref{thm:ucb} if $T$ is unknown.

\begin{theorem}[Restatement of Theorem~\ref{thm:main_informal}]\label{thm:ucb2}
	Given a stochastic bandit instance with unknown $T$, let $\Delta_i = \mu_* - \mu_i$, and let $\Delta = \min_{i:\Delta_i>0} \Delta_i$. For any $T>0$, there exists an algorithm that uses $O(1)$ words of space and achieves regret
	$$O\left(\sum_{i:\Delta_i>0} { \frac{\log (\Delta_i/\Delta)}{\Delta_i}\log T}\right)$$
\end{theorem}
\begin{proof}
	We present the algorithm in Algorithm~\ref{alg:ucb1_unknown_T}.  The algorithm repeatedly calls the procedure in Algorithm~\ref{alg:ucb_known_T} with increasing time horizons $T_0, T_1, \dots, T_L$, where $L \le \log \log T$.  By setting $T_l = T_{l-1}^2$, we have $T_l = T_0^{2^l}$.  Then, by Theorem~\ref{thm:ucb}, we can upper bound the regret as
	\begin{align*}
		\ov{\Psi}_T &\lesssim ~ \sum_{l=0}^{L} \sum_{i=1}^{K} \frac{\log(\Delta_i/\Delta)\log T_l}{\Delta_i} \\
		& = ~ \sum_{l=0}^{L} \sum_{i=1}^{K} \frac{2^l \log(\Delta_i/\Delta)\log T_0}{\Delta_i} \\
		& \lesssim ~ \sum_{i=1}^{K} \frac{2^L \log(\Delta_i/\Delta)\log T_0}{\Delta_i} \\
		& \lesssim ~ \sum_{i=1}^{K} \frac{\log(\Delta_i/\Delta)\log T}{\Delta_i}
	\end{align*}
	which proves the theorem.
\end{proof}

\begin{algorithm}[t]\caption{UCB algorithm with constant space and unknown $T$ (Theorem~\ref{thm:ucb2} and Theorem~\ref{thm:ucb3})}
	\label{alg:ucb1_unknown_T}
	\begin{algorithmic}[1]{
			\Procedure{\textsc{UCBCS-UnknownT}}{$K$}
			\State Initialize $T_0 \gets 10$
			\State $l \gets 0$, $t \gets 1$
			\While { $t \le T$ }
			\State Call \textsc{UCBConstSpace($K,T_l$)}, 
			\State $t \gets t+T_l$
			\State $l \gets l+1$
			\State $T_l \gets T_{l-1}^2$
			\EndWhile
			\EndProcedure}
	\end{algorithmic}
\end{algorithm}

Similarly, we are able to use this trick for the improved algorithm in Section~\ref{sec:ucb3} and get the same regret as in Theorem~\ref{thm:ucb3}.
\begin{theorem}[Restatement of Theorem~\ref{thm:main_informal2}]\label{thm:ucb3_better}
	Given a stochastic bandit instance with unknown $T$, let $\Delta_i = \mu_* - \mu_i$, and let $\Delta = \min_{i:\Delta_i>0} \Delta_i$. For any $\gamma>0$ and any $T>0$, there exists an algorithm that uses $O(1)$ words of space and achieves regret
	\begin{align*}
		O \left( \sum_{i:\Delta_i>0}  \frac{1}{\Delta_i} \left(\log^{\gamma}\frac{1}{\Delta_i} + \frac{\log (\Delta_i/\Delta)}{\gamma \log \log (\Delta_i/\Delta)} \right)\log(T) \right).
	\end{align*}
\end{theorem}

\section{Conclusion} \label{sec:con} We proposed a constant space
algorithm for the stochastic multi-armed bandits problem.  Our
algorithms proceeds by iteratively refining a confidence interval
containing the best arm's value.  In the simpler version of our
algorithm, we refine the interval by a constant factor in each step,
and each iteration only uses $O(\OPT)$ regret.  This gives an
$O(\log\frac{1}{\Delta})$-competitive algorithm.  We then showed how
to improve this by an $O(\log \log \frac{1}{\Delta})$ factor in
certain cases, by using fewer rounds that give more progress.
Finally, we showed how to adapt our algorithms---which involve
parameters that depend on the time horizon $T$---to situations with
unknown time horizon.








%
\section{Acknowledgments}
This work was partially supported by NSF through grants CCF-1216103, CCF-1331863, CCF-1350823 and CCF-1733832.

\appendix

\section{Discussion of Conjecture on the Lower Bound for Stochastic Bandits}
For any given round $r$, for some $\alpha>0$, define:
\begin{align*}
&R_{\inn}^{(r)} = \sum_{i: \Delta_i<\alpha g_r} \frac{1}{g_r},
&R_{\out}^{(r)} =\sum_{i: \Delta_i>\alpha g_r} \frac{1}{\Delta_i}
\end{align*}
That correspond to two terms in 
\begin{align}
&~\sum_{r=1}^{r_i} \Delta_i \cdot \frac{2\log(1/\delta)}{g_r^2} + \sum_{r=r_i+1}^{r_{\max}} \Delta_i \cdot \frac{2\log(1/\delta)}{\left(\Delta_i - g_{r-1}\right)^2} \notag \\
& ~ + \Delta_i \cdot r_{\max}
\end{align}
which is the total regret provided within Section \ref{sec:ucb3}.  Consider the following example where there is a group of high-value arms and a group of low-value arms, and the size of the low-value arms is larger than the high-value arms.
\begin{example}\label{ex:easy1}
	Assume $1>E \gg \epsilon$, and $s>1/2$. Let $\Delta_i=\epsilon$ for $i=1\dots sK$, and $\Delta_i=E$ for $i=sK+1 \dots K$.
\end{example}
In this example, we can find that as $g_r<E/2$, $R_{\inn}^{(r)} = sK/g_r$, and $R_{\out}^{(r)}=(1-s)K/E$.  Since $g_r \lesssim E$ and $s>1/2$, we can find that $R_{\out}^{(r)} \lesssim R_{\inn}^{(r)}$.  This means that Example~\ref{ex:easy1} will not harm us if we use Algorithm~\ref{alg:ucb_known_T} because we know that $\sum_{r} R_{\inn}^{(r)} \lesssim \sum_{i} 1/\Delta_i$.

Then, we consider another example where the size of the group of the high-value arms is larger than low-value arms.  Particularly, we consider
\begin{example}\label{ex:easy2}
	Assume $1>E \gg \epsilon$, and $s<1/2$, where $s/(1-s)<\epsilon/E$. Let $\Delta_i=\epsilon$ for $i=1\dots sK$, and $\Delta_i=E$ for $i=sK+1 \dots K$.
\end{example}
We can find that in this example, as long as $\epsilon \lesssim g_r\lesssim E$, $R_{\inn}^{(r)} \lesssim sK/\epsilon \lesssim (1-s)/E = R_{\out}^{(r)}$.  This means that this is the hard case for Algorithm~\ref{alg:ucb_known_T} because $R_{\out}^{(r)}$ is dominating.  However, we can deal with this example with the following update rule
$$ g_{r+1} = \frac{g_r}{2 \max\{1,(1-s)/s\}} $$
which is roughly $ g_{r+1} = \frac{g_r}{2 \max\{1,R_{\out}^{(r)}/R_{\inn}^{(r)}\}} $.
Note that if $s$ is unknown, we can estimate it by simply counting the number of arms not ruled out.  With the new update rule, we can find that as long as $g_r \lesssim E$, we have $g_{r+1} \lesssim Es/(1-s) \lesssim \epsilon$.  This means that in the next round, we are able to identify the high-value arms.  Therefore, the number of rounds is a constant.

Finally, we consider the following case where we conjectured to be the hard case:
\begin{example}
	Let $\Delta_i = i/K$ for $i = 1,2,\dots, K$.
\end{example}
First note that in this example, $R_{\inn}^{(r)} \eqsim n$ and $R_{\out}^{(r)} \eqsim n\log 1/g_r$, where we can find that $R_{\inn}^{(r)} \lesssim R_{\out}^{(r)}$ for any $r$.  If we use the trick we are dealing with Example~\ref{ex:easy2}, we can find that the corresponding update rule becomes $g_{r+1} = \frac{g_r}{2 \log 1/g_r}$.  Such rule is exactly \eqref{eq:new_update_rule}.  Therefore, we conjecture that the additional $\frac{\log (\Delta_i/\Delta)}{ \log \log (\Delta_i/\Delta)}$ factor is not improvable.

\clearpage 
\bibliographystyle{abbrvnat}
\bibliography{ref}

\end{document}